\pdfoutput=1
\newif\ifFull
\Fulltrue
\ifFull
\documentclass[11pt]{article}
\else
\documentclass{sig-alternate}
\fi
\usepackage{color}
\ifFull
\newenvironment{proof}{\begin{normalsize}\noindent{\bf Proof:}}%
{ $\Box$ \end{normalsize} \\}
\newcommand{\qed}{}
\fi

\usepackage{graphicx}
\usepackage{cite}
\usepackage{times}
\usepackage{latexsym}
\usepackage{amssymb}
\usepackage{amsmath}
\usepackage{hyperref}

\ifFull
\advance \topmargin by -\headheight
\advance \topmargin by -\headsep
\evensidemargin \oddsidemargin
\else 
\fi
     
\makeatletter
\def\@begintheorem#1#2{\sl \trivlist \item[\hskip \labelsep{\bf #1\ #2:}]}
\def\@opargbegintheorem#1#2#3{\sl \trivlist
      \item[\hskip \labelsep{\bf #1\ #2\ #3:}]}
\makeatother

\newcommand{\R}{{\bf R}}

\newtheorem{theorem}{Theorem}
\newtheorem{lemma}[theorem]{Lemma}

\newtheorem{corollary}[theorem]{Corollary}

\newcommand{\UH}{\mbox{\it UH}}

\ifFull \else 

\conferenceinfo{ACM GIS'10,}{November 2--5, 2010, San Jose, CA, USA}
\crdata{978-1-4503-0428-3/10/11}
\CopyrightYear{2010}

\fi

\begin{document}

\title{Privacy-Preserving Data-Oblivious Geometric \\
     Algorithms for Geographic Data}

\date{}

\ifFull
\author{
David Eppstein\\[3pt]
Dept.~of Computer Science \\
Univ.~of California, Irvine \\
eppstein(at)ics.uci.edu
\and
Michael T.~Goodrich \\[3pt]
Dept.~of Computer Science \\
Univ.~of California, Irvine \\
goodrich(at)ics.uci.edu
\and
Roberto Tamassia\\[3pt]
Dept.~of Computer Science \\
Brown University \\
rt(at)cs.brown.edu
}
\else
\numberofauthors{3}
\alignauthor{
David Eppstein\\
\affaddr{Dept.~of Computer Science} \\
\affaddrUniv.~of California, Irvine} \\
\email{eppstein(at)ics.uci.edu}
\and
Michael T.~Goodrich \\[3pt]
Dept.~of Computer Science \\
Univ.~of California, Irvine \\
goodrich(at)ics.uci.edu
\and
Roberto Tamassia\\[3pt]
Dept.~of Computer Science \\
Brown University \\
rt(at)cs.brown.edu
}
\fi

\maketitle 


\begin{abstract}
We give efficient data-oblivious algorithms for several fundamental 
geometric problems that are relevant to geographic information systems,
including planar convex hulls and
all-nearest neighbors.  
Our methods are ``data-oblivious''
in that they don't perform any data-dependent operations,
with the exception of operations performed inside low-level blackbox
circuits having a constant number of inputs and outputs.  Thus, an
adversary who observes the control flow of one of our algorithms, but who cannot
see the inputs and outputs to the blackbox circuits, cannot learn
anything about the input or output.  
This behavior makes our methods
applicable to \emph{secure multiparty computation} (SMC) protocols
for geographic data used in location-based services. 
In SMC protocols, multiple parties wish to perform a
computation on their combined data without revealing
individual data to the
other parties.  
For instance, our methods can be used to solve a problem posed
by Du and Atallah, where Alice has a set, $A$, of $m$ private points in
the plane, Bob has another set, $B$, of $n$ private points in the
plane, and Alice and Bob want to jointly compute the convex hull of
$A\cup B$ without disclosing any more information than what can
be derived from the answer. 
In particular, neither Alice nor Bob want
to reveal any of their respective points that are in
the interior of the convex hull of $A\cup B$.

\end{abstract}


\section{Introduction}
\ifFull\else
\pagestyle{plain}
\def\thepage{\arabic{page}}
\fi

As handheld devices containing GPS receivers have become more popular,
so have location-based services using them.
In particular, the emergence of location-based
mobile social networking services, such as GyPSii,
Pelago, Loopt and Google Latitude, is revolutionizing social networking.
In these applications, the location of a handheld device is 
a critical component of a social-networking
computation, and sometimes is even 
the \emph{sole} attribute of interest with respect to input from the user,
such as for real-time traffic or real-time friend location.
(See Figure~\ref{fig:iphone}.)

\begin{figure}[hbt]
\centering\includegraphics[height=3.8in]{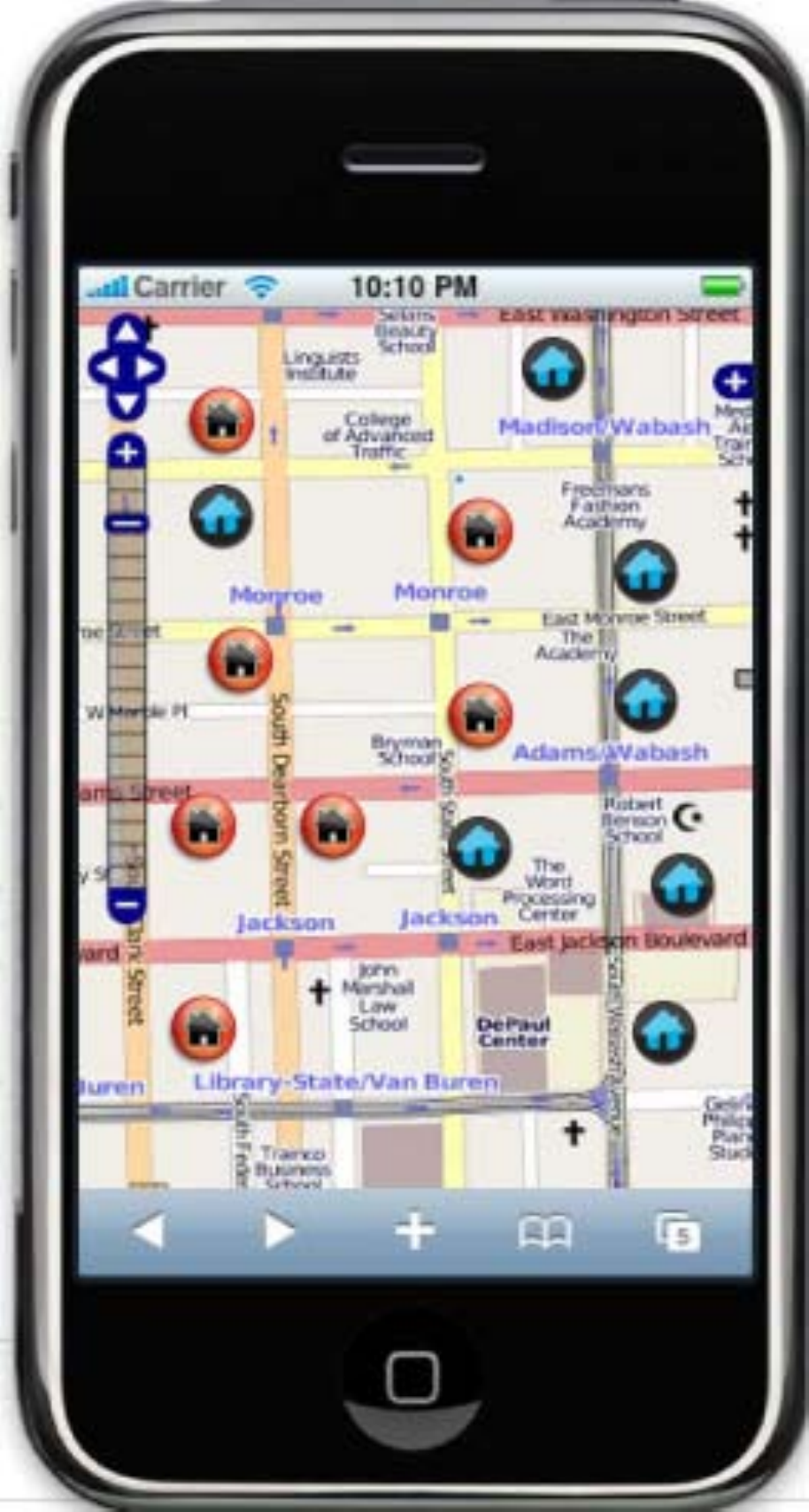}
\caption{\label{fig:iphone} Mock-up of a GPS-based cellphone app
for identifying the locations of people from two different organizations.
(Map image is from openstreetmap.org; public-domain cellphone image
is by Tibounise.)}
\end{figure}

Nevertheless, an individual's physical location is often considered
private information and revealing it to networked applications poses
serious privacy and security risks.
For example, an employee might want to 
conceal from her employer that she is interviewing with a rival
and a husband might want to conceal from his wife where 
he is shopping for her birthday present, not to mention the privacy
concerns associated with trips to a hospital, police station, or
court. 
Even just revealing that one is not at home could be a risk if
that information is discovered by thieves.
Thus, although participating in social location-based services
can have significant benefits,
many users will likely be reluctant to
participate without solid privacy protections.

\subsection{Secure Multi-party Computations}
One way of formalizing privacy requirements for geographic data
is through \emph{secure multi-party computation} (SMC) protocols 
(e.g., see~\cite{bnp-fssmp-08,clos-uctms-02,da-smpcp-01,dz-passm-02,%
m-smpcm-06,mnps-fstpc-04}), 
in which two or more parties 
hold different subsets of a collection of data 
values, $\{x_1,x_2,\ldots,x_n\}$,
and are interested in computing some function,
$f(x_1,x_2,\ldots,x_n)$, on these values.
Due to privacy concerns, none of the 
parties is willing to reveal specific values of his or her
pieces of data.
SMC protocols allow the parties to compute the value of $f$ on their
collective input values without revealing any of their specific data
values (other than what can inferred
from the output value of function~$f$).
One of the main tools for building SMC protocols is to encode the 
function $f$ as a circuit and then simulate an evaluation of this
circuit using 
cryptographically-masked values~\cite{bnp-fssmp-08,mnps-fstpc-04}. 
By unmasking only the output value(s), the parties can
learn the value of $f$ without revealing their own data.
Unfortunately, from a practical standpoint, 
encoding entire computations as circuits can
involve significant blow-ups in space and time.

\subsection{Data-Oblivious Algorithms}
The time and space overhead incurred by SMC protocols
can be managed more efficiently, however, 
by using \emph{data-oblivious algorithms} to drive SMC 
computations~\cite{wlgdz-bpstp-10}.
A data-oblivious computation consists of a sequence of 
data accesses that do not depend on the input values.
All functions that combine data values 
are encapsulated into \emph{black box} operations,
with a constant number of inputs and outputs.
The control flow 
depends only on the input size, the problem being solved, and, in the case of
randomized algorithms, the values of random variables.
A classical example of an oblivious algorithm is a \emph{sorting network}, 
an algorithm that sorts its data values by routing them through black box 
\emph{comparators} that take as input pairs of values and produce as 
output the minimum and maximum of the pair.
However, unlike sorting networks, which are usually described as circuits 
of comparators, we allow data-oblivious algorithms to be structured 
as conventional sequential algorithms using random-access memory, 
looping, and recursion. An algorithm is data-oblivious, in our model, 
if two inputs of the same size have the same distribution of possible
memory accesses.

An adversary
who can see all the control flow of a 
data-oblivious algorithm and all the memory addresses it accesses,
but who cannot see any actual
memory values or the results of any of its black box computations, cannot learn
anything about the specific inputs.
Therefore, in a SMC simulation of a data-oblivious algorithm,
the masked evaluation of each black box may be performed independently,
avoiding the blow-ups that arise from cryptographic simulation 
of non-constant-sized circuits.
The parts of the algorithm outside of the black boxes may be 
performed directly and openly rather than being simulated,
with masked data values taking the place of their unmasked 
values in the algorithm's memory.
The resulting SMC algorithms will be considerably more efficient 
than an SMC simulation of a non-oblivious algorithm.
Our aim, therefore, is to design oblivious
algorithms for geometric problems to be used via 
SMC simulations as components of privacy-preserving location-based services.

\subsection{Problems of Interest}
In this paper, we study several classic geometric problems 
for geographic data, including the following,
for which efficient conventional algorithms can be found 
in computational geometry textbooks
(e.g., see~\cite{bkos-cgaa-97,e-acg-87,o-cgc-98,ps-cgi-90,s-sdsqo-89}):
\begin{itemize}
\item
\emph{convex hull:} given $n$ labeled points in the plane, return the 
labels of those on the boundary of the smallest convex set containing the set
of points.
\item
\emph{quadtree:}
given $n$ points in the plane, construct a representation of the compressed
point-set quadtree~\cite{s-sdsqo-89,s-dasds-90} for this set of points.
\item
\emph{closest pair:}
given $n$ points in the plane, return the pair that are the closest.
\item
\emph{all nearest neighbors:}
given $n$ labeled points in the plane, return, for each point, the label of
its nearest neighbor point.
\end{itemize}
In addition, we study a
more specialized problem---the construction of
a well-separated pair decomposition---which was introduced 
by Callahan and Kosaraju~\cite{ck-dmpsa-95}, 
who also showed how it can be used to solve the all
nearest neighbor problem in a way that generalizes an approach of
Vaidya~\cite{v-oaann-89}.
Chan~\cite{c-wspdl-08} gives a linear-time 
algorithm for computing a well-separated pair decomposition
in the case of integer-coordinates.

\subsection{Related Prior Work}
The general topic of privacy for location-based services
is of considerable interest in GIS (e.g., 
see~\cite{b-pptlbs-09,btgd-spgd-08,cm-pilbs-09,f-plas-09,k-prtlbs-09,l-plalbs-09}).
Of all the problems listed above, the convex hull and 
nearest-neighbor problems are
probably the most well-motivated for geographic data.
For instance, Stojmenovic {\it et al.}~\cite{srl-vdchb-06}
use convex hulls of nearest neighbors
for greedy routing in wireless networks.
Getz and Wilmers~\cite{gw-alnnc-04}
use unions of convex hulls of nearest neighbors to construct species
home ranges from GPS data.
Basch, Guibas, and Hershberger~\cite{bgh-dsmd-97} give data
structures for maintaining convex hulls and closest pairs for mobile
geographic data.
Likewise, at previous ACM~GIS conferences,
Henrich {\it et al.}~\cite{hlb-adgfa-08}
use convex hulls to define geographic footprints for geographic
database queries,
Liu and Lee~\cite{ll-awlun-09} use convex hulls to study 
wireless location using non-line-of-sight radio signals,
and
Buchin {\it et al.}~\cite{bbvl-flspt-09} use convex hulls to
characterize similar parts of trajectories.
In addition, there is considerable prior work
on answering nearest-neighbor queries for both static and mobile GPS data
(e.g., see~\cite{bjks-nrnnq-06,hitkk-dhk-08,pvc-mcnnq-08,rkv-nnq-95}).
None of these prior algorithms is data-oblivious.

In addition to the work on SMC protocols cited above,
Du and Atallah~\cite{da-smpcp-01} survey SMC protocols
and mention several geometric problems, including planar convex
hull, as being of interest for~SMC.
Atallah and Du~\cite{ad-smpcg-01} specifically address privacy-preserving 
computational geometry SMC protocols, including two-party
protocols for point-in-polygon detection, polygon intersection
detection, and closest-pair finding,
although there has been some questions regarding
the correctness of some of these methods\footnote{Du, private communication.}.
Li and Dai~\cite{ld-stpcg-05} study several low-level primitives for
privacy-preserving geometric computations and give a protocol
of complexity $O(n^2)$ for computing the closest red-blue pair
between a set of red points and blue points in the plane.
Wang {\it et al.}~\cite{wlh-pppfc-08} and, independently,
Wang and Zhang~\cite{wz-achap-09}, present SMC protocols for 
two-party convex hull construction, with quadratic communication complexity.
Hans {\it et al.}~\cite{hags-oppch-09} give an improved SMC protocol
for convex hull construction, using a protocol with complexity
$O(n\log n)$, but their method is non-oblivious and 
reveals all the points on the convex hull, 
whereas our method can be used to selectively reveal 
only certain types of points of interest.
In addition, Li {\it et al.}~\cite{lhyz-ptdpp-08} present a
quadratic SMC protocol for approximate three-dimensional convex
hulls.

Goldreich and Ostrovsky~\cite{go-spsor-96} give a general
construction for converting a non-oblivious algorithm into an oblivious one.
Their simulation has an $O(\log^3 n)$ blow-up in
time, which results in inefficient oblivious algorithms
if applied to existing computational geometry algorithms for the
problems we address.

\subsection{Our Results}
We give data-oblivious algorithms for planar convex hull construction, 
well-separated pair decomposition,
compressed quadtree construction,
closest pairs, and all nearest neighbor finding in a set of $n$ points.
Our methods run in $O(n\log n)$ time and, using
known SMC protocols 
(e.g., see~\cite{bnp-fssmp-08,clos-uctms-02,da-smpcp-01,dz-passm-02,%
m-smpcm-06,mnps-fstpc-04}), 
result in privacy-preserving two-party protocols
for performing joint
computations of these geometric algorithms on private data held separately by
two parties,
with communication complexities that
are $O(n\log n)$ times the complexities for the low-level SMC protocols
(used to simulate our low-level blackbox computations.
In addition, we also give oblivious algorithms for list ranking, tree contraction, 
and all nearest larger values, with similar $O(n\log n)$ running times.

Our optimal oblivious convex hull algorithm 
involves the use of a new geometric classification for common tangent finding
for two convex polygons separated by a line, which extends the Overmars and
van~Leeuwen classification~\cite{ol-mcp-81} 
to subsequences of edges of the respective polygons.
Our optimal oblivious algorithms for all nearest neighbors, closest pairs,
compressed quadtree construction,
and well-separated pair decompositions, on the other hand, depend more on
new combinatorial insights than geometric ones, in that our methods are based on
new oblivious computations for list ranking, tree contraction, and all nearest
larger values.


\section{Oblivious Convex Hulls}
Suppose we are given an array $A$ of $n$ points in the plane, 
sorted by their $x$-coordinates (since it
is possible to sort $A$ obliviously in $O(n\log n)$ 
time~\cite{aks-scps-83,g-rsaso-10}).
The desired output is for each point $p$ in $A$ 
to be labeled with a pair of points, $(q,r)$, that form the upper convex hull
edge that is intersected by a vertical line through $p$.
If $p$ is on the upper convex hull, then $(q,r)$ is the convex hull edge that
follows $p$ in the clockwise direction.
This assumption about the output format could also be replaced,
without changing the overall running time, by a compact
listing of the upper hull vertices padded with ``blank'' points so that the
total size is $n$ (since we must maintain the data-oblivious nature 
of our method).
This alternative output format could be produced, for instance, by performing
a compaction operation on the uncompressed output format of labeling each
point with its vertical upper convex hull edge. 
For instance, in a privacy-preserving 
security two-party protocol, Alice could hold 
a set, $A$, of $n$ blue points and Bob
could hold a set, $B$, of $n$ red points, and we could use the 
oblivious convex hull
algorithm we describe in this section to let Alice and Bob each learn
which of their respective points are on the convex hull of $A\cup B$.
(See Figure~\ref{fig:ch}.)

\begin{figure}[hbt!]
\ifFull
\centering\includegraphics[scale=0.5, trim=0in 3in 5.25in 0in, clip]{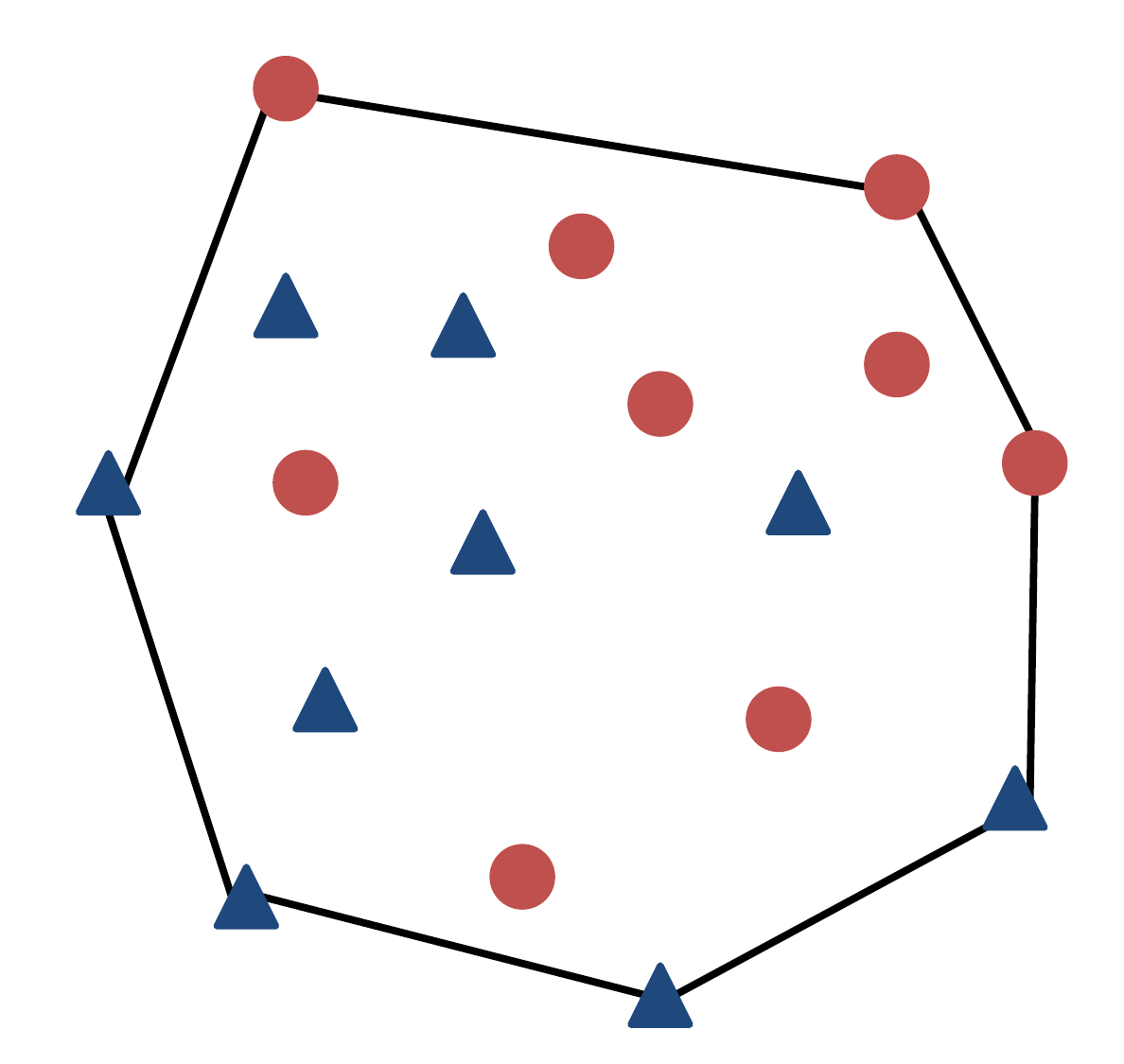}
\else
\vspace*{-8pt}
\centering\includegraphics[width=6.5in]{convexhull}
\vspace*{-2.2in}
\fi
\caption{A two-party convex hull problem. Alice holds a set, $A$, of blue
points (triangles) and Bob holds a set, $B$, of red points (circles). Each
should learn only which of their points are on the convex hull of
$A\cup B$.}
\label{fig:ch}
\end{figure}

\subsection{Background for Our Approach}
Before we present our oblivious convex hull
algorithm, we briefly mention alternative
approaches that do not result in optimal oblivious algorithms.
A standard approach to use the divide-and-conquer paradigm, by 
dividing $A$ into its first and second halves, $A_1$ and $A_2$, and recursively
construct the upper hulls, $\UH(A_1)$ and $\UH(A_2)$,
of the points in $A_1$ and $A_2$, respectively.
The problem that remains is to find the upper tangent segment between
$\UH(A_1)$ and $\UH(A_2)$, and label all the points under this tangent segment 
to have this segment as their upper convex hull edge.
So let us focus on the computation of the upper tangent, $(q,r)$, between 
$\UH(A_1)$ and $\UH(A_2)$.
Since the points in $\UH(A_1)$ and $\UH(A_2)$ are ordered by their 
$x$-coordinates, we can apply a binary search of Overmars and
van~Leeuwen~\cite{ol-mcp-81} to find the upper hull.
The main idea of this method is to probe at two vertices $p\in \UH(A_1)$ and
$q\in \UH(A_2)$ and use the relative positions of the edges next to $p$ and
$q$ to determine which portions of $\UH(A_1)$ and/or $\UH(A_2)$ 
that can be safely
eliminated as candidates for the upper tangent points $q$ and $r$,
respectively.  The case analysis is shown in Figure~\ref{fig:ch-search} and 
results in a running time of $O(\log n)$
for finding the common upper tangent.

\begin{figure}[hbt!]
\vspace*{-2pt}
\centering\includegraphics[width=3.3in]{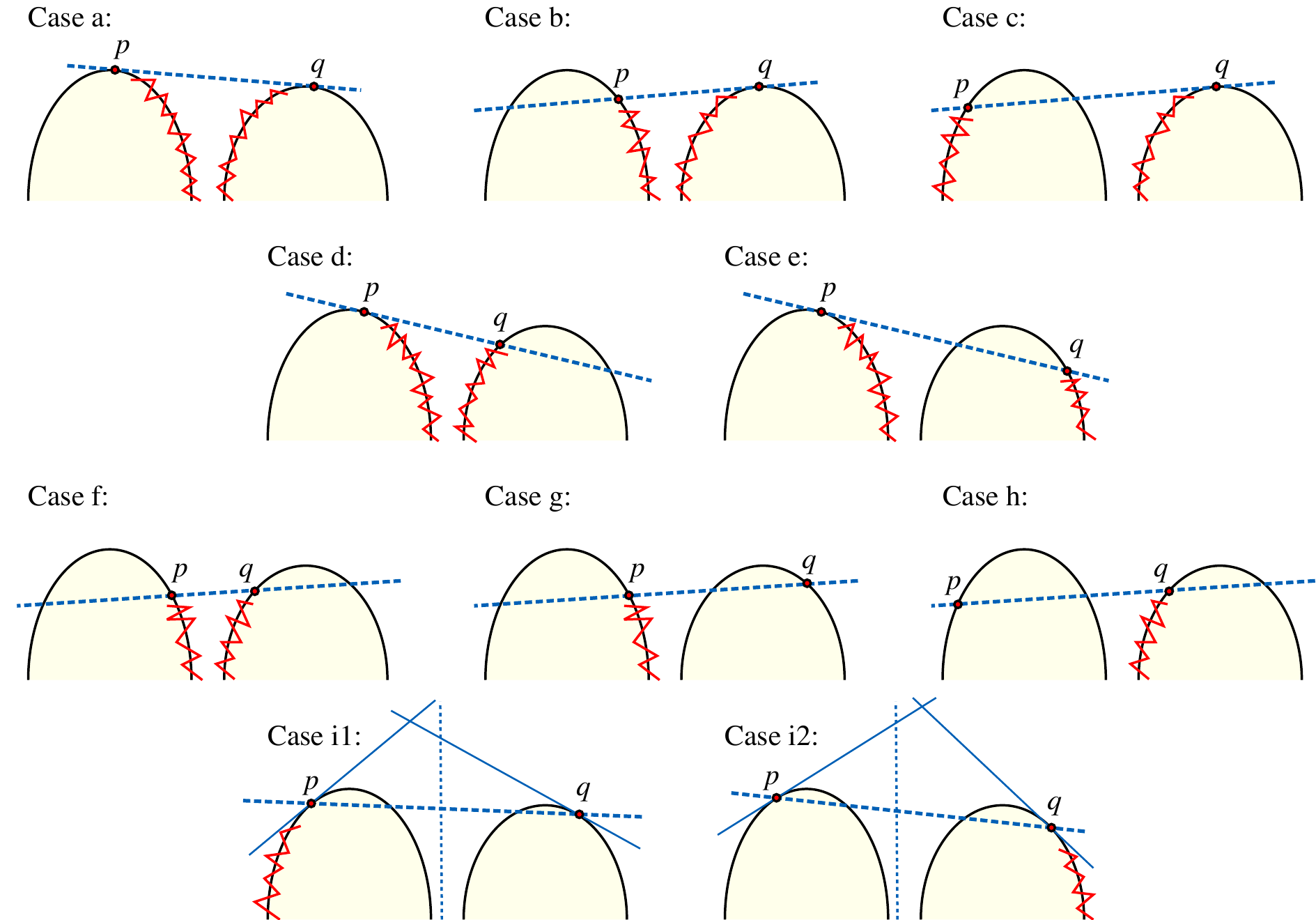}
\vspace*{-10pt}
\caption{The cases for the Overmars-van~Leeuwen binary search for a common
upper tangent. 
Each case shows the relative orientation of two points, which
are respectively on two different upper hulls, and the portion(s) of each
hull that can be eliminated as potential locations for the points of tangency.
}
\label{fig:ch-search}
\end{figure}

Broadcasting this tangent 
to members of $\UH(A_1)$ and $\UH(A_2)$ and comparing
each edge to the tangent allows us to then produce the desired output.
Unfortunately, the binary search process of Overmars and van~Leeuwen is not
oblivious. Moreover, implementing it with the oblivious RAM 
simulation of~\cite{go-spsor-96} blows up the running time
for tangent finding
from $O(\log n)$ to $O(n\log^4 n)$, which results in a running time of
$O(n\log^5 n)$ for oblivious convex hull construction.
\ifFull
Nevertheless, our method borrows from this approach
the idea of using divide-and-conquer.
\fi

An alternative divide-and-conquer approach
is suggested by the parallel convex hull algorithm
of Atallah and Goodrich~\cite{ag-pasft-88}.
In this case, rather than perform a binary search to find the upper tangent,
they perform an $O(n^{1/2})$-way
search in parallel. This $O(n^{1/2})$-way approach is another technique we
borrow, but in a way quite different from that used by
Atallah and Goodrich, as their method is nonoblivious in how it finds upper
tangents.
Implementing their algorithm using a simulation of a PRAM with an oblivious
RAM requires $O(n\log n)$ time to find the upper tangent, which leads to a
running time of $O(n\log^2 n)$ for oblivious convex hull construction.
Instead, our oblivious method is based on a novel geometric
characterizations of when edges of subsequences 
of the two upper hulls come before or after
the tangent line in terms of an ordering by decreasing slopes.
This alternative approach allows us to achieve a running time of $O(n\log n)$
for oblivious convex hull construction.

\subsection{Our Oblivious Convex Hull Method}
Given a set of points, $A$, ordered by their $x$-coordinates,
we define the format of the upper hull, $\UH(A)$, of $A$, to be as follows
For each point $p$ in $A$, we label $p$ with the edge, $e(p)$, of the upper
convex hull that is intersected by a vertical line through the point $p$.
If $p$ is itself on the upper hull, then we label $p$ with the upper hull
edge incident to $p$ on the right.
To simplify the description of our algorithm, we assume
no two points in $A$ share the same $x$-coordinate.

Our method is as follows.
Divide $A$ into its first and second halves, $A_1$ and $A_2$, by a vertical
line $V$ and recursively
construct the upper hulls, $\UH(A_1)$ and $\UH(A_2)$,
of the points in $A_1$ and $A_2$, respectively, with representations as
described above.
In addition,
we assume, without loss of generality, that $\UH(A_1)$ and $\UH(A_2)$ are
each augmented with vertical dummy edges incident on the first and last vertices
in $\UH(A_1)$ and $\UH(A_2)$ respectively.
The problem that remains is to find the upper tangent segment between
$\UH(A_1)$ and $\UH(A_2)$, and label all the points under this tangent segment 
to have this segment as their upper convex hull edge.
So let us focus on the computation of the upper tangent, $(q,r)$, between 
$\UH(A_1)$ and $\UH(A_2)$.

We aim to assign each edge $e$ of $\UH(A_1)$ and
$\UH(A_2)$ one of two labels:
\begin{itemize}
\item
$L$: the tangent line of $\UH(A_1\cup A_2)$ with the same slope as $e$ is
tangent to~$\UH(A_1)$.
\item
$R$: the tangent line of $\UH(A_1\cup A_2)$ with the same slope as $e$ is
tangent to~$\UH(A_2)$.
\end{itemize}
In some intermediate steps, however, we may be unable to determine yet
whether an edge should be labeled $L$ or~$R$; In such cases, we temporarily label it with an~$X$.

If an edge of $\UH(A_1)$ gets label $L$, then it is part
of~$\UH(A_1\cup A_2)$. If it gets instead label $R$, then we know it
is not part of~$\UH(A_1\cup A_2)$. Similar considerations hold for
edges of $\UH(A_2)$ and their labels.
Thus, if we can label each edge in $\UH(A_1)$ as $L$ or $R$,
with no edges labeled $X$, then we can immediately
identify the vertex of tangency on
$\UH(A_1)$---it is the vertex incident on the two edges
respectively labeled $L$ and~$R$.
All edges before this point will be labeled $L$ and all edges after this
point will be labeled~$R$.
Likewise, a similar property holds for~$\UH(A_2)$.
To aid in our characterization of the edges of $\UH(A_1)$ and $\UH(A_2)$, 
we have the following.

\begin{lemma}
\label{lem:test}
Let $H_1$ and $H_2$ be (possibly disconnected) subsequences of the edges of $\UH(A_1)$ and $\UH(A_2)$, respectively, ordered by decreasing
slopes and both containing the dummy vertical edges 
from $\UH(A_1)$ and $\UH(A_2)$ as their respective first and last edges.
For a non-vertical edge $e$ in $H_1$ (resp., $H_2$), 
let $d$ be the edge in $H_2$ (resp., $H_1$) with smallest slope
greater than $e$ and let $f$ be the edge in $H_2$ (resp., $H_1$)
with largest slope less than $e$, noting that $d$ and $f$ are not
necessarily consecutive edges in $\UH(A_2)$ (resp., $\UH(A_1)$).
Then there is a simple comparison rule involving only $d$, $e$, and $f$, with
the result being one of the following outcomes:
\begin{itemize}
\item $e$ is correctly labeled $L$ or $R$.
\item $e$ is labeled $X$, but $d$ is correctly
labeled $L$ and $f$ is correctly labeled $R$.
\end{itemize}
\end{lemma}

\begin{proof}
Without loss of generality, let us assume $e$ is 
in $H_1$ and $d$ and $f$ are in $H_2$, and all three edges have distinct
slopes.
Also, let $V$ be the vertical line separating $A_1$ and $A_2$,
and let $\ell(e)$, $\ell(d)$, and $\ell(f)$ denote the lines containing
$e$, $d$, and $f$, respectively.
We distinguish four cases, with respect to
the points $a=\ell(e)\cap \ell(d)$ and $b=\ell(e)\cap \ell(f)$,
as to whether
(i) $a$ and $b$ are both to the left of $V$,
(ii) $a$ and $b$ are both to the right of $V$,
(iii) $a$ is to the left of $V$ and $b$ is to the right of $V$,
or
(iv) $a$ is to the right of $V$ and $b$ is to the left of $V$.
Note first that case (i) is impossible, since it would require 
the portion of $\ell(d)$ to the right of $V$ be
completely above $\ell(f)$ to the right of $V$ (since $d$ has to be below
$\ell(f)$ to the right of $V$).
For cases (ii), (iii), and (iv),
we illustrate the possibilities in Fig.~\ref{fig:cases}.
For each instance, if there is more than one possible applicable case
according to the Overmars-van~Leeuwen classification (OvL cases a through~i2),
we choose the one that
is the most pessimistic with respect to how much we can determine about the
labels of the respective edges.
In Case~(ii), $e$ is labeled $R$ or $d$ is labeled $L$ and $f$ is
labeled~$R$. In Case~(iii), $e$ is labeled~$R$. Finally, in Case~(iv),
$e$ is labeled~$L$.
\qed
\end{proof}

\begin{figure}[b!]
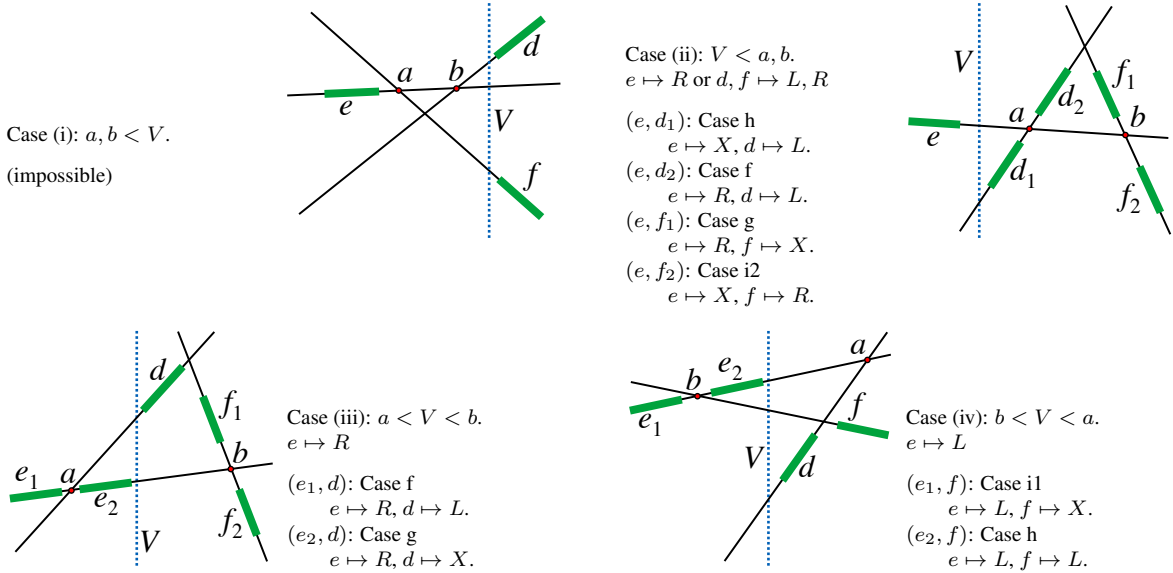

\input cases.tex
\caption{The possible cases for the configurations of $\ell(e)$, $\ell(d)$,
and $\ell(f)$. In some of the cases there are two possible locations for $e$,
$d$, or $f$ relative to the lines containing these edges, in which case we
use subscripts $1$ or $2$ to distinguish the two relative locations. For each
scenario, we list a set of comparisons and their results according to the
OvL classification, together with the resulting
classification of $e$, $d$, and/or~$f$. Note that in Case~(ii)
of particular note is the $(e_2,f)$
comparison in Case~(iv), for this is an example of an OvL Case~h where we can
classify $e_2$ as $L$, since it is impossible for any edges of $\UH(A_2)$ to
be above $\ell(e)$ in this scenario.}
\label{fig:cases}
\end{figure}

Let $N=\lceil\sqrt{n}\rceil$, and let $A'_1$ and $A'_2$ respectively denote
the subarrays of $A_1$ and $A_2$ consisting of the points with indices at
multiples of~$N$.
Let $H_1$ be the subsequence of the recursively constructed upper hull
$\UH(A_1)$ consisting of the edges that have at least one
endpoint vertex in~$A'_1$. Define $H_2$ similarly with respect to
$\UH(A_2)$ and~$A'_2$.
Thus, $H_1$ and  $H_2$ have size~$O(n^{1/2})$.
We perform a \emph{round} in our computation as follows:
\begin{enumerate}
\item
We perform an (oblivious)
brute-force computation to compare every pair of edges in
$H_1\cup H_2$, so as to determine for each non-vertical edge $e$ 
in $H_1$ (resp., $H_2$), 
the edge, $d$, in $H_2$ (resp., $H_1$) with smallest slope
greater than $e$ and $f$, the edge in $H_2$ (resp., $H_1$)
with largest slope less than $e$.
\item
For each
edge $e$, use Lemma~\ref{lem:test} to label $e$ with $L$, $R$, or $X$, as a
blackbox computation applied to each edge, with its associated edges $d$ and
$f$.
\item
Perform another brute-force comparison of every pair of edges 
in $H_1\cup H_2$ to label edges $d$ and $f$, as $L$ and $R$, 
for some edge $e$ whose blackbox computation determined these 
labels for $d$ and $f$.
\item
Perform a forward scan and reverse scan on $H_1$ and $H_2$ using 
a blackbox computation that
labels any edge to the left of an $L$-labeled edge as $L$ and any edge to the
right of an $R$-labeled edge as~$R$.
\end{enumerate}
Note that all of the above steps can be performed obliviously in $O(n)$ time.

\begin{lemma}
After the above round computation completes,
at most one of the subsequences $H_1$ or $H_2$ can contain edges labeled $X$.
\end{lemma}
\begin{proof}
After each application of Lemma~\ref{lem:test} to an edge $e$ in $H_1$ or
$H_2$, we either label $e$
with an $L$ or $R$ or we label $e$ with an $X$ and 
its associated edges $d$ and $f$ as $L$ and $R$.
Note that in the latter case the forward and reverse scans will then
completely label all the edges of the other list (not containing $e$)
as $L$ or $R$; hence, this list will contain no edges labeled $X$.
That is, if we label any edge $e$ as $X$, then we label all the edges in the
other list as $L$ or $R$.
If, on the other hand, we don't label any edge $e$ in $H_1$ (resp., $H_2$) as
$X$, then there are clearly no edges in $H_1$ (resp., $H_2$) that are labeled
as~$X$.
\qed
\end{proof}

Although at most one of $H_1$ or $H_2$ can 
have an edge labeled $X$, the edges in this list may almost all be
labeled~$X$.
Thus the above round computation will reduce the candidate 
tangent vertices in the representation of
one of $\UH(A_1)$ or $\UH(A_2)$ (but not necessarily both) to a subregion of size 
$O(n^{1/2}$).
If it reduces $\UH(A_1)$, call it a \emph{red-1} reduction and
if it reduces $\UH(A_2)$, call it a \emph{blue-1} reduction.
A second application of the round computation will either reduce the other list
to a subregion of size $O(n^{1/2})$ (i.e., it will be a red-1 or blue-1
reduction) or it will reduce the first subregion
to a single vertex of tangency, which we call a \emph{red-2} or \emph{blue-2}
operation, depending on whether it occurs to $\UH(A_1)$ or $\UH(A_2)$.
In either case, two more applications of the round computation will determine
the tangent edge between $\UH(A_1)$ and $\UH(A_2)$.

As described above this sequence of applications of the round computation is
non\-ob\-liv\-i\-ous, but it can be made oblivious by 
considering all valid sequences of red-1, blue-1, red-2, and blue-2, in turn.
One of these constant number of operation sequences 
will be the correct sequence to find the upper tangent. By trying all
these possibilities obliviously (with conditional no-ops for
paths not taken)
we will perform the one 
the leads to the determination of the tangent between $\UH(A_1)$ and
$\UH(A_2)$.

An important implementation detail is the oblivious
method for doing the reduction in a red-1 or
blue-1 operation. A red-2 or blue-2 operation, which reduces a set of
size $O(n^{1/2})$ to an object of size $O(1)$ can be done obliviously
in a single scan using a constant-size register.
In the red-1 or blue-1 operation, we have an array $A$ of size $O(n)$ 
for which we want to isolate a subregion of size $O(n^{1/2})$ based on the labels of its boundary elements, and copy it into a
buffer, $B$, of size $O(n^{1/2})$.
For each subregion, we read the boundary elements and use them to set a register flag, $F$,
that determines whether this is the region that should be copied.
We then read the $i$-th element, $B[i]$, from our buffer, 
perform a conditional swap (based on $F$) with the $i$-th element 
in this region of $A$, and write the result back to $B[i]$.
Thus, in an oblivious way, we can copy a subregion of interest into the buffer
$B$, with the total computation taking $O(n)$ time.

Summarizing, we can perform the determination of the upper tangent of $\UH(A_1)$
and $\UH(A_2)$ obliviously in $O(n)$ time, including the scan of $\UH(A_1)$
concatenated with $\UH(A_2)$ to relabel any vertices under this tangent with
an identifier for this tangent.
A similar construction applies to the
lower hull of $A$. Thus, we have the following result.

\begin{theorem} \label{thm:convex-hull}
Given a set $S$ of $n$ points in the plane, we can 
obliviously construct a representation of
the convex hull of $S$ in $O(n\log n)$ time.
\end{theorem}

Using Theorem~\ref{thm:convex-hull}, we can then apply standard cryptographic circuit simulation
methods to derive a secure multiparty computation involving private data 
(e.g., see~\cite{bnp-fssmp-08,clos-uctms-02,da-smpcp-01,dz-passm-02,%
m-smpcm-06,mnps-fstpc-04}).
Hence, we obtain a secure two-party protocol for Alice and Bob to
determine which of their respective points belong to the convex hull
of the union of their $n$ points with a communication complexity of
$O(n\log n)$.

\begin{corollary} \label{cor:convex-hull} There is a secure two-party
  protocol that computes the convex hull of the union of two private
  sets of points of total size $n$ with $O(n\log n)$ communication
  complexity.
\end{corollary}


\section{Some Combinatorial Problems}
We next turn to oblivious algorithms for some combinatorial problems that
crop up in our methods for our other geometric algorithms.

\subsection{All Nearest Larger Values}
In the \emph{All Nearest Larger Values} (ANLV) 
problem~\cite{bsv-odlpa-93}, we are given an array
$A$ of $n$ numbers, such that, for each value, $A[i]$,
we want to determine the values $A[j]$ and $A[k]$, where 
$j$ is the largest index less than $i$ with $A[j]>A[i]$ and
$k$ is the smallest index greater than $i$ with $A[k]>A[i]$.
As observed by Berkman {\it et al.}~\cite{bsv-odlpa-93},
this problem is actually a generalization of the problem of merging two
sorted lists, $C$ and $D$, since these lists can be merged by
solving an ANLV problem for an array that consists of 
a reversal of $C$ followed by $D$.
Our oblivious method for solving the ANLV problem,
where we assume without
loss of generality that the values are distinct, 
is different from that of Berkman {\it et al.}~and is as follows.
\begin{enumerate}
\item
Build a complete binary tree, $T$, ``on top'' of the items in $A$
and perform a bottom-up tournament computation to compute, for each
$v$ in $T$, the value, $M(v)$, which is the maximum value stored in a
descendent of $v$ in $T$.
This is a straightforward oblivious computation that takes $O(n)$ time.
\item
For each leaf $x$ in $T$,
let $l(x)$ denote the lowest node in $T$ 
such that $l(x)$ is a left sibling of an ancestor of $x$ in $T$ and
$M(l(x))>A[x]$, where $A[x]$ is the value in $A$ associated with $x$.
Likewise, let $r(x)$ be denote the lowest node in $T$ 
such that $r(x)$ is a right sibling of an ancestor of $x$ in $T$ and
$M(r(x))>A[x]$, where $A[x]$ is the value in $A$ associated with $x$.
We compute $l(x)$ and $r(x)$, which are initially null for each $x$,
in a divide-and-conquer computation, with respect to a node $v$ in $T$.
In this computation,
we recursively compute the $l$ and $r$ labels for nodes in the subtrees
rooted at $v$'s left and right children, $u$ and $w$, producing lists, $D(u)$
and $D(w)$, of labeled descendents of $u$ and $w$.
Then, we scan $D(u)$ to assign
$r(x)=w$ for each $x$ such that $A[x]<M(w)$ and $r(x)$ was previously null. 
Also, we scan $D(w)$ list to assign
$l(x)=u$ for each $x$ such that $A[x]<M(u)$ and $l(x)$ was previously null. 
We then concatenate these two lists of labeled nodes (some of which are still
null) to create the list $D(v)$ for $v$, which is passed up to $v$'s parent.
This step runs in $O(n\log n)$ time.
\item
For each node $v$ in $T$ with left child $u$ and right child $w$,
we perform a scan of $D(u)$ to find the smallest 
value $A[x]$ such that $A[x]>M(w)$, if it exists, and we scan back through
$D(w)$ to label the value $A[y]$ in $w$'s list such that $A[y]=M(w)$
to show that $A[x]$ (if it exists)
is the nearest larger value to the left of $A[y]$.
Likewise,
we perform a scan of $D(w)$ to find the smallest 
value $A[x]$ such that $A[x]>M(u)$, if it exists, and we scan back through
$D(u)$ to label the value $A[y]$ in $u$'s list such that $A[y]=M(u)$
to show that $A[x]$ (if it exists)
is the nearest larger value to the right of $A[y]$.
(These scans are used to take care of the boundary values for each node.)
Finally, when we are done with all the scans, we perform a sorting 
step to report back to each node the labels that have been found for it.
This step runs in $O(n\log n)$ time.
\item
For each element $x$ stored in a leaf of the binary tree $T$, 
let us create two tuples, 
$(l(x),A[x],{\rm ``right''},-i,L,R)$ 
and
$(r(x),A[x],{\rm ``left''},i,L,R)$,
where $i$ is the index of the value $A[x]$ in $A$, and $L$ and $R$ are
the left and right ANLV's for $A[x]$ (most of which are
probably null at this point).
Perform an oblivious sort of all these tuples, using a lexicographic
ordering rule, to produce the sorted list, $B$, of such tuples.
This step takes $O(n\log n)$ time.
\item
Scan the list $B$ in reverse order. During this scan we 
maintain three registers,
$v$, $L$ and $R$. The register $v$ is a label of the current node, $v$, for
which we are computing ANLV's for, that is, the first coordinate of the
tuples we are scanning.
The scan for each $v$ is essentially a merge of its left and right children's
lists of nodes whose ANLV is determined by this merge at $v$.
The register $L$ is maintained to be
the smallest ``left'' value in a tuple with this $v$ as
its first ($r(x)$) coordinate.
The register $R$ is maintained to be
the smallest ``right'' value in a tuple with this $v$ as
its first ($l(x)$) coordinate.
Whenever we encounter a tuple, if it is a ``right'' tuple, we identify 
its left ANLV as $L$, and if it is a ``left'' tuple, we identify
its right ANLV as $R$, assuming we have not already determined this value
previously (which coincides with the point where we reset the register $v$).
This scan can be done obliviously in $O(n)$ time.
\item
Perform one more sort to bring together the computed 
left and right ANLV's for each node $x$ in $T$.
This step can be done obliviously in $O(n\log n)$ time, and it completes the
algorithm.
\end{enumerate}
Thus, we have the following.

\begin{theorem}
Given an array $A$ of $n$ values, we can obliviously 
solve the ANLV problem for $A$ in $O(n\log n)$ time.
\end{theorem}

\subsection{List Ranking}
In the \emph{list ranking}
problem~\cite{am-dplr-88,cv-dctao-86},
we are given a linked list, $L$, stored in the
records of an array of size $n$, for which we want to compute, for each node
$v$ the number of nodes from $v$ to the end of the list, using pointer
hopping as the distance metric.

\begin{theorem}
\label{thm:list-rank}
Given a linked list $L$ of $n$ nodes, we can obliviously perform a list
ranking of the nodes in $L$, using a computation that always runs in 
$O(n\log n)$ time and succeeds with high probability.
\end{theorem}

\begin{proof}
To solve the list ranking problem on a list $L$ of $n$ nodes
obliviously, we perform the following actions:
\begin{enumerate}
\item
Create for each node $v$ in $L$ a field, $d(v)$, which stores an indication
of the distance from $v$ to the end of $L$.
Initially, each $d(v)=1$.
\item
Generate a random bit, $b(v)$, for every node $v$ in $L$.
\item
Perform two oblivious sorts so as to ``link out'' each node $w$ 
with $b(w)=0$ that follows a node $v$ with $b(v)=1$, storing with $w$ a
reference to name of the node, $u$, that currently follows $w$ in $L$.
In addition, with this link out step, we update $d(v)=d(v)+d(w)$.
\item
Repeat the previous two steps a constant, $c$, number of times, until it is
likely, with high probability, that the connected part of $L$
has at most half as many nodes.
\item
If $|L|>n/log n$,
perform an oblivious sorting
step and compression to reduce the working storage used for $L$ to be half
its previous size.  Then repeat the above computation starting with Step~2.
\item
If $|L|\le n/log n$, then perform $O(\log n)$ link-out steps, where we apply
a link-out operation, like the one above, for every node in parallel, using
oblivious sorts to perform the actions in an oblivious fashion.
The total running time for all these 
actions is $O(((n/\log n)\log n)\log n)=O(n\log n)$.
\item
Reverse the above link-out steps, in reverse order, so that,
for each node $w$ that was linked out in step $i$, 
we perform two oblivious sorting steps
to communicate the information needed so that we can update $d(w)=d(w)+d(u)$,
where $u$ was the node that followed $w$ when it was linked out.
\end{enumerate}

Since we reduce the number of nodes, and the working storage for $L$, 
by half every
$c$ steps, with high probability, and we then reverse these actions 
to finally solve the
list ranking problem, we get that the running time of this method
is a geometric sum that is $O(n\log n)$.
Moreover, since we terminate the halving process and switch to a parallel
link-out process when $|L|\le n/\log n$, we get that this
method succeeds in computing a list ranking for $L$ with high probability.
\qed
\end{proof}

\subsection{Tree Contraction}
In a
\emph{tree contraction}~\cite{adkp-sptca-89,mr-ptcia-85,rmm-lrptc-93}
computation, we are given a proper binary tree $T$ 
such that each leaf node
is associated with a value and each internal node is associated with an
arithmetic
operation to be computed on its two children.
The goal is to efficiently 
compute the value of each node in $T$ in an oblivious fashion, 
even if the height of $T$ is $O(n)$.

\begin{theorem}
\label{thm:tree}
Given a binary arithmetic tree, $T$, with $n$ nodes,
we can obliviously compute the value of each internal node of $T$ in
$O(n\log n)$ time, in a computation that succeeds with high probability.
\end{theorem}

\begin{proof}
Adapting a parallel algorithm of Abrahamson {\it et al.}~\cite{adkp-sptca-89},
we can solve the tree contraction problem obliviously as follows.
\begin{enumerate}
\item
Perform a list ranking operation to number the leaves of $T$ from $1$ to $N$.
Using the algorithm described below, this step takes $O(n\log n)$ time and succeeds
with high probability.
\item
For each node $v$ in $T$ that is an odd-numbered leaf 
and a left child of its parent, link
out $v$ and its parent, making $v$'s sibling to be the new child of $v$'s
grandparent. In doing this operation, record for $v$ and its parent
the iteration it is removed and the names of the grandparent and sibling
nodes at this point.
In addition, in the link-out operation, we label the child-parent edge with
an $O(1)$-sized algebraic operation to apply in going from the child value to
the parent (which is composable when we combine previously-computed edges in
a link-out).
This step can performed obliviously using $O(1)$ oblivious sorting steps.
\item
For each node $v$ in $T$ that is an odd-numbered leaf 
and a right child of its parent, link
out $v$ and its parent, making $v$'s sibling to be the new child of $v$'s
grandparent. In doing this operation, record for $v$ and its parent
the iteration it is removed and the names of the grandparent and sibling
nodes at this point, along with any edge updates as in the previous step.
This step can performed obliviously using $O(1)$ oblivious sorting steps.
\item
If $|T|>1$,
divide the leaf number of each leaf node by $2$ and repeat the above two
steps.
\item
Reverse the above actions to compute the value of each internal node.
\end{enumerate}
This completes the proof.
\qed
\end{proof}


\section{Quadtrees and Well-Separated Pair Decompositions}
Having described our methods for some fundamental combinatorial problems,
we describe in this section our oblivious algorithms for 
constructing a compressed quadtree and for forming a well-separated pair
decomposition.

\subsection{Constructing a Compressed Quadtree}
A \emph{compressed quadtree}~(e.g., 
see~\cite{bet-pcqqt-99,c-wspdl-08,s-sdsqo-89}), 
for a set $S$ of $n$ points in the
plane, normalized to have the unit square, $[0,1]\times [0,1]$,
as a bounding box, is
defined as follows.
A \emph{quadtree} (e.g., see~\cite{s-sdsqo-89}) 
for $S$
is defined recursively, where we create a node $v$ for the 
current bounding box and, if this bounding box has more than a given
threshold number of points of $S$, 
then we divide this box into four equal-sized boxes 
as quadrants, and we recursively construct subtrees for each non-empty quadrant,
with the nodes for these non-empty quadrants having $v$ as their parent.
(See Figure~\ref{fig:quadtree}.)
If we then compress any chains of nodes in this quadtree
that have only one child, then we get the compressed quadtree.
This definition is clearly not something that leads to an oblivious
construction algorithm, of course, but we can in fact
construct a compressed quadtree for $S$ obliviously in $O(n\log n)$ time.

\begin{figure}[hbt]
\centering\includegraphics[width=3.3in]{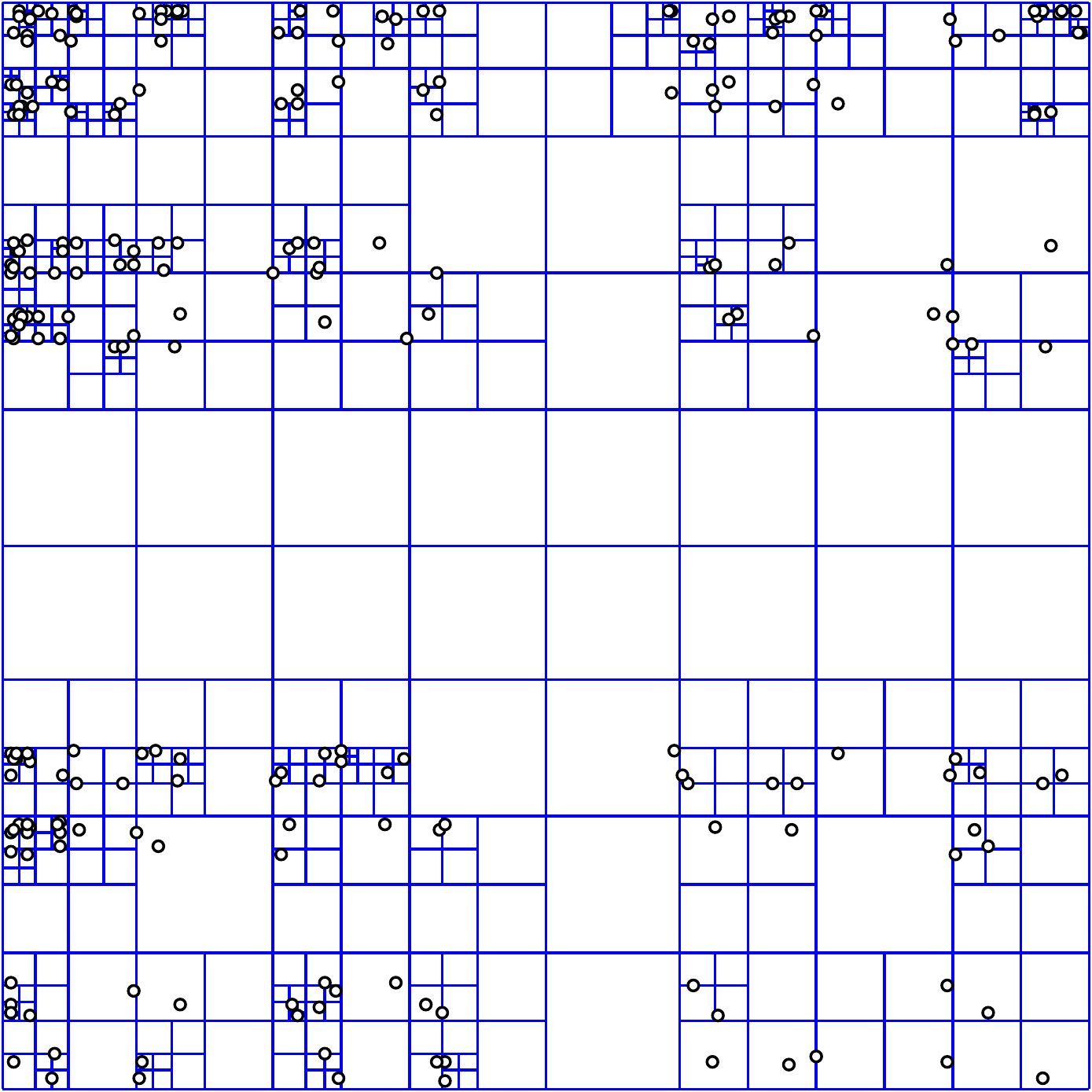}
\caption{\label{fig:quadtree} A (region) quadtree for a set of points.
(Public-domain image by David Eppstein.)}
\end{figure}

An alternative method for constructing a compressed quadtree, as
observed by several researchers (e.g.,
see~\cite{bet-pcqqt-99,c-wspdl-08,s-sdsqo-89}), 
is based on a sorting of the points of $S$ according to the 
\emph{interleaving order}.
In the interleaving order, we take each point $(x,y)$ and interleave the bits
for $x$ with the bits for $y$ is a standard shuffling, and we compare points
according to this order. This order can also be interpreted
geometrically~\cite{c-wspdl-08} for the sake of a comparison-based sorting
algorithm.
Once we have the points of $S$ stored in an array $A$ according to the
interleaving order, we note, as shown by 
Bern {\it et al.}~\cite{bet-pcqqt-99}, that the nodes
contained in any compressed quadtree box 
form a contiguous subsequence in $A$.
Moreover, we can label each transition between two adjacent points in $A$
with the box that is formed along that transition, and we can then identify
the compressed quadtree box that contains each point $p$ in $A$ by performing
an ANLV computation, where we use box size to determine values in
this ANLV computation.
Given this information, we can perform a postprocessing 
step consisting of two oblivious sorting
steps to determine the adjacency information between the parent and child
nodes in the compressed quadtree.
Thus, we have the following.

\begin{theorem}
\label{thm:quadtree}
Given a set $S$ of distinct points in the plane, we can obliviously construct a
compressed quadtree for $S$ in $O(n\log n)$ time.
\end{theorem}

\subsection{Well-Separated Pair Decomposition}
Another important geometric computation is
the construction of a well-separated pair decomposition (WSPD) for a set $S$
of $n$ points in the plane.
In a WSPD~\cite{ck-dmpsa-95}, 
we are given a parameter $s$ for
which we want to construct a set of pairs,
$(A_1,B_1),(A_2,B_2),\ldots,(A_k,B_k)$, 
such that every pair of points $p$ and $q$ are represented by a pair
$(A_i,B_i)$ such that $p\in A_i$ and $q\in B_i$, or vice versa,
and such that there are balls of radius $r$ containing $A_i$ and $B_i$
respectively so that these balls are of distance at least $sr$ apart.
In our applications we choose $s>2$ to be a constant, e.g., $s=2.1$ will do.

We should note that some authors also insist 
that each pair of points $p$ and $q$
be represented exactly once in some $(A_i,B_i)$ pair in WSPD.
But duplicate
representation is actually not a problem for most WSPD applications,
including the ones we consider, so we don't make this additional
requirement. What is essential in our definition 
is that the total number of pairs in a WSPD be linear.

Given a compressed quadtree for a set $S$ of $n$ points
in the plane, Chan~\cite{c-wspdl-08}, shows that a WSPD
can be constructed in $O(n)$ time by a simple
(non-oblivious) recursive search algorithm defined on the nodes of the
compressed quadtree. 
Using a technique of Callahan and Kosaraju,
Chan shows that the time and combinatorial complexity for his algorithm is 
$O(n)$ by using a packing argument, which shows that the number of compressed
quadtree boxes that are no smaller than a box $b$ but are too close to be
candidates for a well-separated pair with $b$ is bounded by a constant
depending on $s$.

We define an alternative, oblivious construction algorithm for a
well-separated pair decomposition by turning this construction and argument
``on its head.''
That is, we use the packing argument itself to construct the WSPD.
In particular, for each box $B$ in the compressed quadtree, $T$, there are
$O(s^2)=O(1)$ boxes in the uncompressed quadtree, $T'$,
that are the same size as $B$ and are not well-separated from $B$. And for
each such box, $B'$, there are $O(1)$ immediate (children and grandchildren)
descendents of the edge in $T$ corresponding to where $B'$ is located in
$T'$.
These immediate descendents and the children of $B$ in $T$
together form candidates for well-separated pairs. And the collection 
of all such sets of candidate pairs form a superset of the pairs that are
considered by the WSPD algorithm of
Chan~\cite{c-wspdl-08}.
Thus, if we can consider all such pairs and only keep the ones that form
well-separated pairs, then we can construct a WSPD
of size~$O(n)$.

The challenge is to collect all such pairs.
We do this as follows.
\begin{enumerate}
\item
For each box $B$ in the compressed quadtree, $T$,
form the set, ${\cal B}(B)$ of $O(s^2)=O(1)$ boxes 
that are the same size as $B$ in the uncompressed quadtree and are not
well-separated from $B$.
\item
In parallel,
for each $B$ in $T$, pick a box $B'$ in ${\cal B}(B)$ that was previously
unconsidered, 
and create two points $p_B$ and $q_B$ inside $B'$ that do not fall
inside the same child box of $B'$ in the uncompressed quadtree, $T'$.
Call these points ``dummy points.''
\item
Create a compressed quadtree, $\hat T$, for all the points in $S$ together
with the dummy points created in the previous step. Note that the
box $B'$ is in $\hat T$, even if it is not in $T$.
\item
Label each point in $S$ with a $1$ and each dummy point with a $0$ and
perform an tree compression on $\hat T$, where each internal-node operation
is a binary OR, to determine the binary value of
each internal node in $\hat T$. Note that the nodes of $\hat T$ that also
exist in $T$ will have at least two children that have binary values equal to
$1$. 
\item
Remove all the nodes with binary values equal to $0$ from $\hat T$ and
construct an Euler tour of its edges, perform a list ranking in that Euler
tour, and then an ANLV computation on the nodes in this list using node
degree as the item values.
This computation gives us, for each box $B'$ in a ${\cal B}(B)$ set, the
highest nodes in $T$ whose boxes are contained in $B'$.
These nodes and their children, together with the children of $B$, form
candidates for well-separated pairs. Identify which ones are indeed
well-separated and compress them into a list of answers produced in this
round.
\item
Repeat Steps~2 through 5 above until we have considered each 
box in a set ${\cal B}(B)$, for its box $B$.
\end{enumerate}

Each of the above steps runs in $O(n\log n)$ time, with the list ranking and
tree evaluation steps succeeding with high probability. Likewise, there
are only $O(1)$ iterations to this algorithm. So we get the following.

\begin{theorem}
\label{thm:wspd}
Given a set $S$ of $n$ points in the plane, and a compressed quadtree $T$
for $S$, we can construct a
well-separated pair decomposition for $S$, with each set being
associated with a node in $T$, in $O(n\log n)$ time with an
oblivious computation that succeeds with high probability.
\end{theorem}


\section{Closest Pairs and All Nearest Neighbors}
Having presented all the above algorithmic techniques, we are now
ready to describe our oblivious algorithm for solving the all nearest
neighbors problem.

So, let us assume we are given a set $S$ of $n$ points in the plane
for which we want to solve the all nearest neighbors problem.
At a high level it is an oblivious adaptation of a parallel all
nearest neighbors algorithm of 
Callahan and Kosaraju~\cite{ck-dmpsa-95}.
\begin{enumerate}
\item
Construct a compressed quadtree $T$ for
the points of $S$, using the method of Theorem~\ref{thm:quadtree}.
\item
Construct a well-separated pairs decomposition (WSPD), based on $T$,
using the method of Theorem~\ref{thm:wspd}.
\item
Discard each pair $(A,B)$ in the WSPD if neither $A$ nor $B$ is a
singleton set.
(Note: if all we want is a closest pair, then we can skip the remaining steps
and find the closest of all the singleton-singleton pairs in the WSPD.)
\item
For each box $B$ in the WSPD for which there is at least one remaining pair,
$(\{a\},B)$, construct the set $N(B)$ of all such points, $a$. We represent this information obliviously as a collection of pairs $(a,B)$ where $a$ is a point and $B$ is a box, padded with null items.
\item
For each box $B$, partition the plane into a set of $O(1)$ wedges having the center $o$ of $B$ as their apex, and prune $N(B)$ to contain only the closest point to $o$ within each wedge (by replacing the pairs representing other points by null items), with the number of wedges chosen according to the parameters of the WSPD so that for each pruned point there is another point of $N(B)$ that is closer to it than every point in $B$ is. The set of remaining points
in each set $N(B)$ will have size $O(1)$.
\item
Using a tree contraction algorithm of 
Callahan and Kosaraju~\cite{ck-dmpsa-95}, 
construct for each leaf $v$ of $T$, which is
associated with a point $b$, the set $N'(b)$, which is the set of 
all points $a$ such that $(a,B)$ is in the WSPD for an ancestor of
$v$ in $T$ and such that $a$'s distance to $b$ is no larger than the minimum
distance from $a$ to other points in $N(B)$.
In other words, $N'(b)$ consists of all those points of $S$ that
could have $b$ as a nearest neighbor.
This step takes $O(n\log n)$ time to implement obliviously, by
Theorem~\ref{thm:tree}.
\item
For each point $a$ in a set $N'(b)$, construct the pair $(a,b)$. Sort
all these pairs so as to bring together, for each point $a$, those
points that could be a nearest neighbor to $a$. Then perform a scan
of this list to determine, for each $a$, its nearest neighbor.
This step takes $O(n\log n)$ time.
\end{enumerate}
Each of the above steps can be implemented in $O(n\log n)$ time,
either because of the specific results from the referenced theorems,
or because the step is easily performed obliviously by making a
constant number of calls to an oblivious sorting routine.

\begin{theorem} \label{thm:neighbors}
Given a set $S$ of $n$ points in the plane, we can compute the
nearest neighbor in $S$ for each point in $S$ and find a closest
pair of points in $S$ with an oblivious
computation running in $O(n\log n)$ time.
\end{theorem}

Starting from the data-oblivious algorithms of Theorem~\ref{thm:neighbors},
we can then apply standard cryptographic circuit simulation
methods to derive a secure multiparty computation involving private data 
(e.g., see~\cite{bnp-fssmp-08,clos-uctms-02,da-smpcp-01,dz-passm-02,%
m-smpcm-06,mnps-fstpc-04}).
Hence, we obtain secure two-party protocol for Alice and Bob to
compute either the closest pair or 
the nearest neighbor in the union of their $n$ points for
each of their respective points, but otherwise learn nothing about
the other person's points. 

\begin{corollary} \label{cor:neighbors}
  There is a secure two-party protocol that computes the all nearest
  neighbors and a closest pair in the union of two private sets of
  points of total size $n$ with $O(n\log n)$ communication complexity.
\end{corollary}

The result of Corollary~\ref{cor:neighbors} is perhaps counter-intuitive,
in that one might, at
first, believe that such a computation reveals
all of the points in question. 
However, if Alice and Bob's respective sets of points are relatively
well-separated, then each of them would learn almost nothing from a
two-party all nearest neighbors computation, for, in this case, each
of their respective points has a nearest neighbor 
in its same original set.


\section{Conclusion}
We have given efficient oblivious algorithms for a number of
geometric problems, which are natural problems to arise in
privacy-preserving protocols for computing functions of points that
are derived from the coordinates of actors in various location-based
services.
We have also given oblivious algorithms for several fundamental combinatorial
problems.
There are a host of open problems, however, that might be
of interest in privacy preserving computations, including the following:
\begin{itemize}
\item
Given a set $S$ of $n$ vertical and horizontal
line segments, can one obliviously compute in $O(n\log n)$ time
the number of pairs of segments in $S$ that intersect?
\item
Given a set $S$ of $n$ points in the plane, can one construct a
representation of the Delaunay triangulation of $S$ obliviously in
$O(n\log n)$ time?
\item
Given a set $S$ of $n$ points in $\R^3$, can one construct a
representation of the convex hull of $S$ obliviously in
$O(n\log n)$ time?
\item
Given a simple polygon $P$ of size $n$, can one construct a
representation of a triangulation of $P$ obliviously in $O(n\log n)$
time? If so, is this the fastest time possible for an oblivious
algorithm?
\end{itemize}


\subsection*{Acknowledgments}
We would like to thank Wenliang (Kevin) Du for several stimulating
discussions regarding the connections between oblivious algorithms
and privacy-preserving protocols.
This research was supported in part by NSF grants 
0830149, 0830403, 1011840, and 1012060
and by ONR under MURI grant N00014-08-1-1015.


{\raggedright
\bibliographystyle{abbrv}
\bibliography{goodrich,k_anonymity,geom,extra}
}

\end{document}